\documentclass[11pt]{article}
\usepackage{fullpage,graphicx,psfrag,amsmath,amsfonts,verbatim}
\usepackage{xcolor}
\usepackage{amsthm}
\usepackage[small,bf]{caption}
\usepackage{authblk}

\newcommand{\reals}{{\mathbb{R}}}

\DeclareMathOperator{\ima}{im}
\DeclareMathOperator{\Div}{div}

\newcommand{\Vol}{\mathop{\bf vol}}

\newtheorem{thm}{Theorem}[section]
\newtheorem{lem}{Lemma}[section]
\newtheorem{cor}{Corollary}[section]

\newtheorem{asmp}{Assumption}[section]
\newtheorem{defn}{Definition}[section]

\newtheorem{example}{Example}[section]

\def\approxcorrect{\checkmark\kern-1.1ex\raisebox{.89ex}{$\times$}}

\usepackage{amsmath,amsfonts,bm}

\def\eqref#1{equation~\ref{#1}}

\def\1{\bm{1}}

\DeclareMathAlphabet{\mathsfit}{\encodingdefault}{\sfdefault}{m}{sl}
\SetMathAlphabet{\mathsfit}{bold}{\encodingdefault}{\sfdefault}{bx}{n}

\def\gF{{\mathcal{F}}}
\def\gG{{\mathcal{G}}}

\def\gN{{\mathcal{N}}}

\def\gP{{\mathcal{P}}}

\def\gU{{\mathcal{U}}}
\def\gV{{\mathcal{V}}}

\allowdisplaybreaks

\title{On the Decomposition of Differential Game}
\author{Nanxiang Zhou}
\author{Jing Dong}
\author{Yutian Li}
\author{Baoxiang Wang \\ \texttt{\{nanxiangzhou,jingdong\}@link.cuhk.edu.cn, \{liyutian,bxiangwang\}@cuhk.edu.cn}}
\affil{The Chinese University of Hong Kong, Shenzhen}
\date{}
\begin{document}
\maketitle

\begin{abstract}
To understand the complexity of the dynamic of learning in differential games, we decompose the game into components where the dynamic is well understood. One of the possible tools is Helmholtz's theorem, which can decompose a vector field into a potential and a harmonic component. This has been shown to be effective in finite and normal-form games. However, applying Helmholtz's theorem by connecting it with the Hodge theorem on $\reals^n$ (which is the strategy space of differential game) is non-trivial due to the non-compactness of $\reals^n$. Bridging the dynamic-strategic disconnect through Hodge/Helmoltz's theorem in differential games is then left as an open problem \cite{letcher2019differentiable}. In this work, we provide two decompositions of differential games to answer this question: the first as an exact scalar potential part, a near vector potential part, and a non-strategic part; the second as a near scalar potential part, an exact vector potential part, and a non-strategic part. We show that scalar potential games coincide with potential games proposed by \cite{monderer1996potential}, where the gradient descent dynamic can successfully find the Nash equilibrium. For the vector potential game, we show that the individual gradient field is divergence-free, in which case the gradient descent dynamic may either be divergent or recurrent. 
\end{abstract}


\maketitle

\section{Introduction}
One of the most fundamental questions in game-theoretic learning is whether an uncoupled learning dynamic can ultimately achieve a stable equilibrium through repeated interactions among players. Specifically, in which games and under what conditions can players reach a stable state using learning algorithms such as gradient descent? This issue has gained significant attention, particularly as many advancements in machine learning have relied on gradient descent to optimize the parameters of neural networks, with objective functions that model non-cooperative games. Popular examples include adversarially generative networks \cite{goodfellow2020generative}, federated learning \cite{kairouz2021advances,donahue2021optimality}, multi-agent reinforcement learning \cite{littman1994markov,lowe2017multi}, and any machine learning algorithm trained in an adversarial way. 

A notable result by \cite{hart2003uncoupled,hart2006stochastic} presents a negative finding, demonstrating that no uncoupled learning dynamics can converge to a Nash equilibrium in all games from any initial condition. This raises the critical question of which games a learning process can successfully converge to a Nash equilibrium and which games it cannot.

In the case where the game is finite, or if the game is a normal-form game, this question can be partially answered by decomposing the game through Helmholtz's theorem. When the number of strategies for each player is finite, or when the strategies considered are on a probability simplex, \cite{candogan2010projection,legacci2024geometric} showed that a game can be decomposed into a potential game and a harmonic game. This decomposition categorizes finite games along a spectrum, ranging from players with fully aligned interests (represented by games containing only the potential component) to players with entirely conflicting interests (represented by games containing only the harmonic component). In potential games, where there exists a potential function to quantify how individual strategy changes affect collective utility, players can thus descent along the direction of the gradient of their utility functions, which is equivalent to collectively minimizing the potential function, to reach the Nash equilibrium. In normal form games, the harmonic game is shown to be incompressible, hence implying that common learning dynamics, such as the exponential weight, can lead to Poincar\'e recurrence \cite{legacci2024geometric}. 

In differential games, where the strategies are assumed to be in $\reals^n$, applying Helmholtz's theorem becomes non-straightforward. Different from the normal form games, where the utilities are multilinear and the strategies are naturally in a compact set, the utility of differential games can be much more complicated.  Specifically, Helmholtz's theorem operates in $\reals^3$, where the curl of the gradient field is still a vector field, which means naively applying Helmholtz's theorem only gives a decomposition in $\reals^3$. When $n \ge 4$, the curl of the gradient field is no longer a vector field, and one then needs to leverage the Hodge Theorem to perform the decomposition on a manifold. Connecting and applying Hodge/Helmholtz decomposition on the manifolds is yet to be investigated. A direct sum decomposition, like those in finite games and normal-form games, remains open in differential games \cite{letcher2019differentiable}. 


\subsection{Related works}
In finite games, where the number of strategies is finite for each player, \cite{candogan2011flows} introduced a method to decompose a given game into a potential and harmonic component. This decomposition maps finite games into a spectrum of players having fully aligned interests (games with only the potential component), to players having completely conflicting interests (games with only the harmonic component). Follow-up works then develop different variants of decompositions for the finite games \cite{cheng2016decomposed,wang2017weighted,li2019note,abdou2022decomposition}. 

This decomposition framework is then extended to normal form games \cite{legacci2024geometric}, where classic no-regret algorithms such as exponential weights are known to be possibly chaotic \cite{palaiopanos2017multiplicative,mertikopoulos2018cycles,vlatakis2019poincare}. Based on the decomposition of normal form games, \cite{legacci2024geometric} provided a principled way to identify cycling behaviors of exponential weights. 

In differential games, \cite{letcher2019differentiable} identified two classes of games based on the symmetric and skew-symmetric parts of the game's Jacobian matrix. The games with the symmetric Jacobian matrix are identified to be potential games, while games with skew-symmetric Jacobian matrix are named Hamiltonian games, which are closely related to harmonic games. However, this is different from the direct sum decomposition results in finite and normal-form games. Specifically, given a game with individual gradient field $g$, it is impossible in general to find a potential game $g_p$ and a Hamiltonian game $g_h$ such that $g = g_p + g_h$. Classic gradient descent methods are known to be convergent for potential games, but they can be non-convergent for Hamiltonian games. They thus introduced symplectic gradient adjustment to find stable fixed points in differential games. They also remarked that connecting the differential-geometric Hodge/Helmholtz decomposition in differential games is left as an open problem. 

\subsection{Differential games}
We consider a differential game with $M$ players. Each player $i$ has utility $\left\{u_i: \mathbb{R}^n \rightarrow \mathbb{R}\right\}_{i=1}^M$ and can play a strategy $\omega_i \in \mathbb{R}^{n_i}$ to maximize its utility. We denote the joint strategy as $\omega=\left(\omega_1, \ldots, \omega_M\right) \in \mathbb{R}^n$ where $\sum_{i=1}^M n_i=n$. We also let $\omega_{-i}$ be the joint strategy of all players except for player $i$. To denote the individual components of $\omega$, we write $\omega = (\omega_1, \ldots, \omega_M) = (x_1, \ldots, x_n)$.

\subsection{Our contributions}
We identified two different inner products on the space of differential games, which allows us to apply the Helmholtz decomposition on the vector field space of the utility gradient. Similar to the decomposition of the finite games, we decompose the differential games into three parts, which enjoy different dynamic properties. However, different from the case of finite games, the differential games cannot be decomposed straightforwardly into a potential part and a harmonic part. This is due to the fact that the harmonic component is isomorphic to the de Rham cohomology of the manifold, which is zero when the differential $k$-form is with $k = 1$ and the manifold is $\reals^n$. Instead, we identify a Vector Potential part of the differential games, which are similar to a harmonic game in many ways, such as the challenges it imposes on dynamical systems induced by gradient descent. 

Our two decompositions provide different interpretations of the space of differential games. In the first decomposition, the game is decomposed into an exact scalar potential part, a near vector potential part, and a non-strategic part. We show that the exact scalar potential part of the game is an exact potential game, and the vector potential part poses similar challenges to learning algorithms. Specifically, the standard gradient descent dynamic can exhibit non-convergent behaviors on the vector potential game. To relate the vector potential games and the Hamiltonian games, we show that a Hamiltonian game has to be a vector potential game, but not vice versa. In the second decomposition, the game is decomposed into a near scalar potential part, an exact vector potential part, and a non-strategic part. We show that the vector potential part of the game is flat on any local Nash equilibrium, which imposes significant challenges to first and second-order local Nash equilibrium finding algorithms. 



\section{Main results}
\subsection{Overview of results}
Before we present the technical results, we first provide an overview of our main result. In this work, we provide two decomposition methods for differential games. In the first decomposition, the game can be decomposed into an exact scalar potential game and a near vector potential game. In the second decomposition, the game can be decomposed into a near-scalar potential game and an exact vector potential game. We show that the scalar potential game coincides with the canonical notion of potential game, where there exists a potential function to quantify how individual strategy changes affect collective utility. We call it the scalar potential game as the curl of the vector field is zero. For the other game, we took the name vector potential game from physics, where vector potential refers to a vector field with zero divergence. The game has not been defined before. In fact, in finite games, because the combination curl of the gradient is zero, the vector potential game does not exist.    The two decompositions can be summarized by the following informal claim. 
\begin{thm}[Informal statement combining Theorem \ref{thm:decompose_C1} and Theorem \ref{thm:decompose_C1_tilde}]
    The simultaneous gradient $Du = (\nabla_{\omega_1} u_1, \ldots, \nabla_{\omega_M} u_M)$ can be decomposed as
    \begin{align*}
        Du = \ & X_\gP + X_\gV \\
        = \ & X_{\tilde{\gP}} + X_{\tilde{\gV}} \,,
    \end{align*}
    where $\gP$ is the class of exact scalar potential game, $\gV$ is the class of near vector potential game, $\tilde{\gP}$ is the class of near scalar potential game and $\tilde{\gV}$ is the class of exact vector potential game. 
\end{thm}

For the vector potential game, we show that the vector field of the gradient is divergence-free. As an implication of that, algorithms based on the gradient descent dynamic may be hard to converge.
Specifically, we show that without a specific initialization requirement, the gradient descent dynamic can drift to infinity, in which case the orbits are not bounded, or the induced trajectory may be chaotic. In this case, while it is possible to pick a good initialization point, there is no principle on how to pick such a point in general. This thus provides insights into the difficulty of learning differential games, as we describe in the below claim.
\begin{thm}[Informal statement combining Theorem \ref{thm:poincare_near_vector} and Theorem \ref{thm:poincare_vector}]
    When a game is in $\gV$ and the utility function has compact support, the gradient descent dynamic is Poincar\'e recurrent. When a game is in $\tilde{\gV}$, and the orbits are bounded, the gradient descent dynamic is Poincar\'e recurrent. When the dynamic is Poincar\'e recurrent, for almost every initialization, the induced trajectory $x(t)$ returns arbitrarily close to the initial point infinitely often.
\end{thm}

\subsection{Notations and background}
Let $L_{\mathrm{loc}}^1(\reals^n)$ denote the set of locally integrable functions on $\reals^n$.  Let $C^\infty(\reals^n)$ denote the space of infinitely differential functions $f: \reals^n \rightarrow \reals$. We let the subscript $c$ denote all the compactly supported functions, i.e. $C^\infty_c(\reals^n)$ denotes the space of compactly supported, infinitely differential functions on $\reals^n$.

In the context of infinite-dimensional normed linear spaces, Hilbert spaces possess a direct sum decomposition related to the inner product. Therefore, to ensure that our considered spaces are Hilbert spaces (which are complete inner product spaces), we need to use the concept of weak derivatives. 

\begin{defn}[Weak derivative, \cite{evans1998partial}]
Suppose $u, v \in L_{\mathrm{loc}}^1(\reals^n)$ and $\alpha$ is a multiindex. We say that $v$ is the $\alpha^{\text {th}}$-weak partial derivative of $u$, written
$$
D^\alpha u=v,
$$
provided
$$
\int_{\reals^n} u D^\alpha \phi d x=(-1)^{|\alpha|} \int_{\reals^n} v \phi d x
$$
for all $\phi \in C_c^{\infty}({\reals^n})$.
\end{defn}

When the derivative of $f$ exists, it is equivalent to the weak derivative. The introduction of weak derivatives is solely to guarantee the completeness of our inner product space. To get an intuitive understanding of the weak derivative, we use the following example as an illustration. 
\begin{example}
    Define $f(x) \in L^2(\reals^n)$,
    \begin{align*}
    f(x) = \begin{cases}
         0 & x < -1 \\
         x+1 & -1 \leq x \leq 0 \\
         -x+1 & 0 < x \leq 1 \\
         0 & x > 1 
         \end{cases}\,, \quad 
        f^\prime(x) = \begin{cases}
         0 & x < -1 \\
         1 & -1 \leq x \leq 0 \\
         -1 & 0 < x \leq 1 \\
         0 & x > 1 
         \end{cases}
    \end{align*}
    Notice that the derivative of $f$ is not defined at $0$, but the weak derivative $f^\prime(x)$ can be defined.
    \begin{figure}[h]
        \centering
        \includegraphics[width=0.5\linewidth]{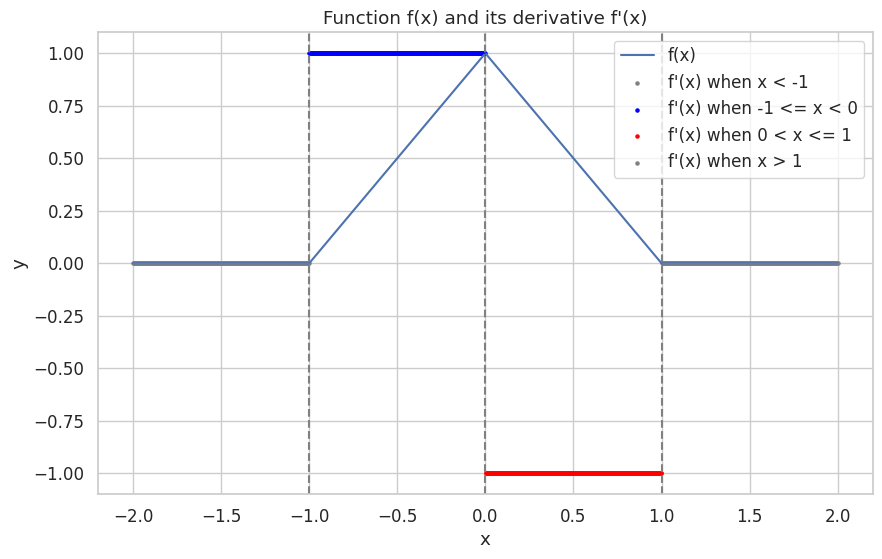}
        \caption{Illustration of the example function and its weak derivative. }
        \label{fig:weak_der}
    \end{figure}
\end{example}

\begin{defn}[Sobolev Space, \cite{evans1998partial}]\label{defn:sobolev}
    The Sobolev space $H^{k} (\reals^n)$ consists of all locally summable functions $f: \reals^n \rightarrow \reals$ such that for each multiindex $\alpha$ with $\alpha \leq k$, $D^\alpha u $ exists in the weak sense and belongs to $L^2(\reals^n)$. Moreover, if $f \in H^{k} (\reals^n) $, we define its norm to be 
    \begin{align*}
        \|f\|_{H^{k} } = \left(\sum_{\alpha \leq k} \int_{\reals^n} |D^\alpha f|^2 dx\right)^{1/2} \,,
    \end{align*}
    where $D^\alpha f$ is the $\alpha$th-weak partial derivative of $f$.  
\end{defn}

\subsection{Decomposition of differential Games}
We follow a similar decomposition roadmap to that of finite games and normal form games \cite{candogan2011flows,legacci2024geometric}. However, different from differential games, the utility functions of normal-form games are multilinear and the strategy space is compact, which makes it significantly easier to decompose. To complete our decomposition, we provide a description of the gradient vector field of the utility function, which is with a newly introduced inner product and norms. 

To begin with, we first introduce three spaces, $C_0, C_1, C_2$, and operators $d_1, d_2$. The domains and codomains of $d_1$ and $d_2$ are summarized as follows.
\begin{align*}
    C_0 \overset{d_1}{\rightarrow} C_1 \overset{d_2}{\rightarrow} C_2 \,.
\end{align*}
We then provide a decomposition of $C_1$, which captures the description of the gradient vector field of the game, based on the Helmholtz decomposition.

We first define space $C_0$ as follows. 

\begin{defn}\label{defn:C0}
    The Hilbert space $C_0$ is defined as $C_0 = \overline{C^\infty(\reals^n)}^{\|\cdot\|_{H^2}}$, which is the closure of the space $C^\infty(\reals^n)$ with norm $\|\cdot\|_{H^2}$. The inner product is defined as 
    \begin{align*}
        \langle f, g\rangle_{_{0}} 
        = \ & \int_{\reals^n} f(x) g(x) dx + \sum^n_{i=1} \int_{\reals^n} \frac{\partial f}{\partial x_i} \cdot \frac{\partial g}{\partial x_i} dx \\
        & \ + \sum^n_{i=1} \sum^n_{j=1}\int_{\reals^n} \frac{\partial^2 f}{\partial x_i \partial x_j} \cdot \frac{\partial g}{\partial x_i \partial x_j}  dx \,.
    \end{align*}
    The norm is thus $\|f\|_{_{0}} = \left[\langle f, f \rangle_0\right]^{1/2}$.
\end{defn}
To ensure a simpler notation, we use $x_i$ to denote the integrals and derivatives with respect to the $i$-th component of the joint strategy. We remark that $\frac{\partial f}{\partial x_i}$ is the weak derivative with respect to $f$ on each component of possible joint strategies $\omega$. 

\begin{defn}\label{defn:C1}
    The Hilbert space $C_1$ is defined as $$C_1 = \left\{X = (f_1, \ldots, f_n) \mid f_i \in H^{1} (\reals^n)\right\}\,.$$ The inner product is defined as 
    \begin{align*}
        \langle X, Y \rangle_{_{1}} = \sum^n_{i=1} \int_{\reals^n} f_i \cdot g_i dx + \sum^n_{i=1}\sum^n_{j=1}\int_{\reals^n}  \frac{\partial f_i}{\partial x_j} \cdot \frac{\partial g_i}{\partial x_j}dx \,. 
    \end{align*}
    The norm is thus $\|X\|_{_{1}} = \left[ \langle X, X \rangle_1\right]^{1/2}$.
\end{defn}
We define the operator that maps between $C_0$ and $C_1$ as $d_1$. 
\begin{defn}\label{defn:d_1}
    The operator $d_1: C_0 \rightarrow C_1$ is defined as follows. For $f \in C_0$, $$d_1 f = \left(\frac{\partial f}{\partial x_1}, \ldots, \frac{\partial f}{\partial x_n}\right)\,.$$
\end{defn}

Lastly, we define the space $C_2$ and the operator $d_2$. 
\begin{defn}\label{defn:C2}
    The Hilbert space $C_2$ is defined as $$C_2 = \left\{\gF \bigm| (\gF)_{ij} = \begin{cases} 
      f_{ij} & i < j \\
      0 & i \geq j
   \end{cases} \,, f_{ij} \in L^2\left(\reals^n\right), \forall i, j \in [n]\right\}\,.$$ 
   The inner product is defined as $$\langle \gF, \gG \rangle_{_{2}} = \sum_{i < j} \int_{\reals^n} f_{ij} \cdot g_{ij} dx \,.$$ The norm is thus $\|\gF\|_{_{2}} = \left[ \langle \gF, \gF \rangle_2\right]^{1/2}$.
\end{defn}

\begin{defn}\label{defn:operator2}
    The operator $d_2: C_1 \rightarrow C_2$ is defined as follows. For $X = (f_1, \ldots, f_n)$, $$(d_2 X)_{ij} = \begin{cases}
        -\frac{\partial f_i}{\partial x_j} + \frac{\partial f_j}{\partial x_i} &  i < j \\
        0 & i \geq j \,.
    \end{cases}$$
\end{defn}

We now connect these definitions to the class of differential games. We first consider differential games with a class of utility function $u \in C_0^M$, where $$C_0^M = \{u = (u_1, \ldots, u_M) \mid u_i \in C_0 \cap C^\infty (\reals^n), \forall i\} \,.$$ We then define the operator $D$, such that 
\begin{align*}
    D_m u = \ & (0, \ldots, \nabla_{\omega_m} u_m, \ldots) \,, \\
    Du = \ &  \sum^M_{m=1} D_m u = ( \nabla_{\omega_1} u_1, \ldots,  \nabla_{\omega_M} u_M)\,.
\end{align*}
Then, we first show that the gradient vector fields are in $C_1$. 

\begin{lem}\label{lem:Du_in_C1}
    For any $m \in [M]$, we have $Du \in C_1$, i.e. $(Du)_m \in H^1(\reals^n)$.
\end{lem}
\begin{proof}
     As $u_m \in C_0$, by the definition of the inner product, we have
     \begin{align*}
         \|u_m\|_0^2 = \ & 
         \int_{\reals^n} |u_m|^2 dx + \sum^n_{i=1} \int_{\reals^n} \left|\frac{\partial u_m}{\partial x_i}\right|^2 dx \\
         & \ + \sum^n_{i=1} \sum^n_{j=1} \int_{\reals^n} \left|\frac{\partial^2 u_m}{\partial x_i \partial x_j}\right|^2 dx \\
         < \ & \infty \,.
     \end{align*}
     Therefore, $\left\|\frac{\partial u_m}{\partial x_i}\right\|^2_{L^2} < \infty$,$\left\|\frac{\partial^2 u_m}{\partial x_i \partial x_j}\right\|^2_{L^2} < \infty$. Hence $Du \in C_1$. 
\end{proof}

To obtain a decomposition of $C_1$, it remains to show that $d_2$ is a bounded linear operator, which is guaranteed by leveraging the inner products defined on $C_0$, $C_1$. 
\begin{lem}\label{lem:d2_bounded}
    $d_2$ is a bounded linear operator. 
\end{lem}
\begin{proof}
    By the definition of $d_2$, it is clearly linear. To show that it is bounded, we need to show that there exists some constant $c$ such that $\|d_2X\|_2 \leq c \|X\|_1$. Using the inequality of $(a-b)^2 \leq 2a^2 + 2b^2$, we have
    \begin{align*}
        \|d_2X\|_{_{2}}^2 
        = \ & \sum_{i < j} \int_{\reals^n} \left|-\frac{\partial f_i}{\partial x_j} + \frac{\partial f_j}{\partial x_i}\right|^2 dx \\
        \leq \ & 2 \left(\sum_{i < j} \int_{\reals^n} \left|\frac{\partial f_i}{\partial x_j} \right|^2 + \left|\frac{\partial f_j}{\partial x_i}\right|^2 dx\right) \\
        \leq \ & 4 \|X\|_1^2 \,,
    \end{align*}
    where the last inequality is by noticing the second term of the definition of $\|\cdot\|_1$. Taking the square root of both sides shows that $d_2$ is bounded. 
\end{proof}

\begin{thm}\label{thm:decompose_C1}
    $C_1$ can be decomposed as 
    \begin{align*}
        C_1 = \ker(d_2) \oplus \ker(d_2)^\perp \,.
    \end{align*}
    Then $$Du = X_\gP + X_\gV\,,$$ where $X_\gP = \ker (d_2)$, $X_\gV = \ker(d_2)^\perp$. 
\end{thm}
\begin{proof}
    As $d_2$ is bounded, we have $\overline{\ker(d_2)} = \ker(d_2)$. Then, we obtain the result by noticing  $C_1 = \overline{\ker(d_2)} \oplus \ker(d_2)^\perp$. By Lemma \ref{lem:Du_in_C1}, we can decompose $Du$ as described.  
\end{proof}

Now we can define the following classes of subgames, which can be interpreted as the space of scalar potential games, near vector potential games, and non-strategic games. 
\begin{align*}
    \gP = \ & \left\{u \in C_0^M \mid Du \neq 0\,, Du \in \ker(d_2)\right\}\cup \{0\} \,, \\
    \gV = \ & \left\{u \in C_0^M \mid Du \neq 0\,, Du \in \ker(d_2)^\perp\right\} \cup \{0\} \,, \\
    \gN = \ & \left\{u \in C_0^M \mid Du = 0\right\} \,.
\end{align*}
It is easy to see that the intersection between any two classes is $0$. 

We first discuss the games that differ by only a non-strategic subgame. Then, we discuss the games with no near vector potential parts or no scalar potential parts in detail. 

As implied by its name, the non-strategic games are a class of subgames that do not affect the set of equilibria. Formally, we define the strategic equivalence of differential games. 
\begin{defn}
    Let $\gG$ and $\gG^\prime$ be two different differential games that only differ in the utility functions. Let $\{u\}_{i=1}^{M}$ be the utility functions of $\gG$ and $\{u^\prime\}_{i=1}^{M}$ be the utility functions of $\gG^\prime$. $\gG$ and $\gG^\prime$ are said to be strategically equivalent, if for any $i \in [M]$, 
    \begin{align*}
        u_i(\omega_i^\prime, \omega_{-i}) - u_i(\omega_i, \omega_{-i}) = u_i^\prime(\omega_i^\prime, \omega_{-i}) - u_i^\prime(\omega_i, \omega_{-i}) \,,
    \end{align*}
    for any strategy $\omega_i, \omega_i^\prime, \omega_{-i}$. 
\end{defn}
Strategically equivalent games have the same payoff comparisons per player, and hence have the same set of Nash equilibria. The non-strategic set of $\gN$ helps characterize strategically equivalent games through the following lemma. 
\begin{lem}
    Two differential games are strategically equivalent if their difference is a non-strategic game. 
\end{lem}
\begin{proof}
    The proof is almost identical to that of \cite{candogan2011flows,legacci2024geometric} and hence we omit it here. 
\end{proof}

\subsection{The scalar potential game}
One of the most popular solution concepts in multi-agent settings is Nash equilibrium, which originates from game theory to describe the behaviors of rational, selfish players \cite{roughgarden2010algorithmic}. It characterizes a stable state between the players, where no individual has any incentive to unilaterally deviate from their chosen strategy.  It is worth noting that the computational complexity of finding a Nash equilibrium in a general game is known to be PPAD-Hard \cite{chen2009settling,daskalakis2013complexity}. Nonetheless, it has been established that the computation of a Nash equilibrium becomes more feasible in specific game contexts, such as potential games \cite{monderer1996potential}, where a potential function is available to quantify how individual strategy changes affect collective utility. A long line of works has developed efficient algorithms that can converge to Nash equilibrium in potential games \cite{cominetti2010payoff,chen2016generalized,heliou2017learning,cuilearning,panageas23,anagnostides2022last}. The formal definition of a potential game is given as follow. 
\begin{defn}[\cite{monderer1996potential}]\label{defn:potential}
A game is a potential game if there is a single potential function $\phi: \mathbb{R}^n \rightarrow \mathbb{R}$ and positive numbers $\left\{\alpha_i>0\right\}_{i=1}^M$ such that
$$
\phi\left(\omega_i^{\prime}, \omega_{-i}\right)-\phi\left(\omega_i^{\prime \prime}, \omega_{-i}\right)=\alpha_i\left(u_i\left(\omega_i^{\prime}, \omega_{-i}\right)-u_i\left(\omega_i^{\prime \prime}, \omega_{-i}\right)\right) \,,
$$
for all $i$ and all $\omega_i^{\prime}, \omega_i^{\prime \prime}, \omega_{-i}$. 
\end{defn}

We now show that if a game is a scalar potential game (without a vector potential part), then it has a utility function classified in $\gP \oplus \gN$ and is an exact potential game.

\begin{lem}\label{lem:exact_then_closed}
    If $X \in \ker(d_2)$, we say it is closed; if $X \in \ima(d_1)$, we say it is exact. Then, every exact $X$ is closed, i.e. $d_2 d_1 = 0$. 
\end{lem}
\begin{proof}
    By the definition of $d_1$, we have $d_1 f = \left(\frac{\partial f}{\partial x_1}, \ldots, \frac{\partial f}{\partial x_n}\right)$. Then applying $d_2$ yields 
    \begin{align*}
        (d_2 d_1 f)_{ij} = - \frac{\partial^2 f}{\partial x_j x_i} + \frac{\partial^2 f}{\partial x_i x_j}  = 0 \,. \tag*{\qedhere}
    \end{align*}
\end{proof}

\begin{lem}[Poincar\'e Lemma]\label{lem:pointcare}
   In $\reals^n$, every smooth closed 1-form is exact with $n>1$. 
\end{lem}

\begin{lem}\label{lem:potential1}
    If $u \in \gP \oplus \gN$, then there exist some $\phi \in C^\infty(\reals^n)$, such that $d_1 \phi = Du$. 
\end{lem}
\begin{proof}
    By the definition of $\gP$, $\gN$, we have $Du \in \ker(d_2)$ and is therefore closed. Since $u \in C_0^M$, $Du$ is also smooth. Thus by the application of Lemma \ref{lem:pointcare} and Lemma \ref{lem:exact_then_closed}, we have the result. 
\end{proof}

\begin{lem}\label{lem:potential2}
    If $u \in \gP \oplus \gN$, then the Jacobian metric $$J(\omega) = \left(\begin{array}{cccc}
\nabla_{\omega_1}^2 u_1 & \nabla_{\omega_1, \omega_2}^2 u_1 & \cdots & \nabla_{\omega_1, \omega_n}^2 u_1 \\
\nabla_{\omega_2, \omega_1}^2 u_2 & \nabla_{\omega_2}^2 u_2 & \cdots & \nabla_{\omega_2, \omega_n}^2 u_2 \\
\vdots & & & \vdots \\
\nabla_{\omega_M, \omega_1}^2 u_M & \nabla_{\omega_M, \omega_2}^2 u_M & \cdots & \nabla_{\omega_n}^2 u_M 
\end{array}\right)$$ is symmetric, and $u$ is an exact potential game. 
\end{lem}
\begin{proof}
    By Lemma \ref{lem:potential1}, it is easy to see that the Jacobian matrix is symmetric. Then by Lemma 2 and Corollary 3 of \cite{letcher2019differentiable}, we have the equivalence to an exact potential game.
\end{proof}

While it is long known that the gradient descent dynamic is convergent in scalar potential games, we provide the Figure \ref{fig:potential} of the gradient descent dynamic to contrast with the dynamic in vector potential games (which is discussed in the next section). 
\begin{figure*}[t]
    \centering
    \includegraphics[width=\linewidth]{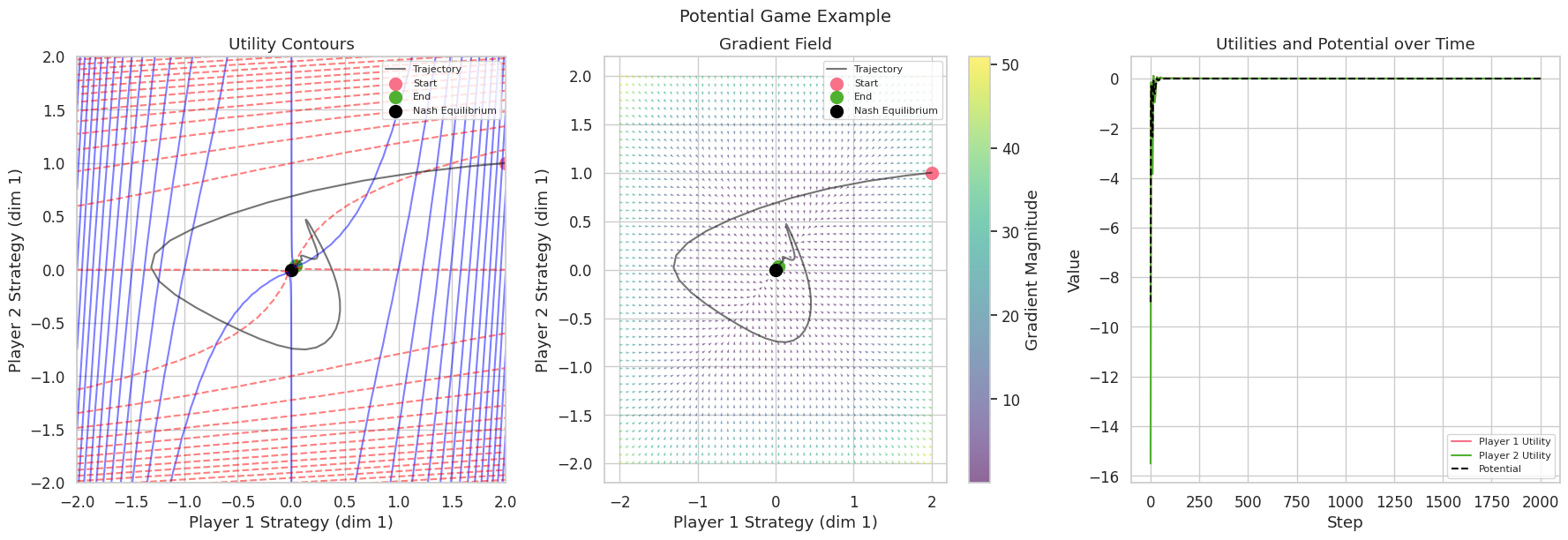}
    \caption{In this two-player potential game, the utility functions are given by $u_1(x,y) = -\frac{1}{2}\|x\|^2 + x^Ty - \|x\|^4$ for player 1 and $u_2(x,y) = -\frac{1}{2}\|y\|^2 + x^Ty - \|y\|^4$ for player 2, where $x,y \in \mathbb{R}^2$. These utilities are derived from the potential function $P(x,y) = -\frac{1}{4}(\|x\|^2 + \|y\|^2) + \frac{1}{2}x^Ty - \frac{1}{4}(\|x\|^4 + \|y\|^4)$. The unique Nash equilibrium is at $(0,0)$.}
    \label{fig:potential}
\end{figure*}
\subsection{The near vector potential game}
Now we turn our attention to the near vector potential game, which is a game with a divergence-free gradient field. We took the name vector potential from physics, where it refers to a field that is divergence-free. The divergence of $X = (g_1, \ldots, g_n)$ is defined as $\Div(X) = \sum^n_{i=1} \frac{\partial g}{\partial x_i}$. 
\begin{lem}\label{lem:div0_ae}
     If $X \in (\ker(d_2))^\perp$, $X = (g_1, \ldots, g_n)$ and each $g_i \in H_c^3(\reals^n) := \overline{C^\infty_c(\reals^n)}^{\|\cdot\|_{H^3}}$, then $\Div(X) = 0$ almost everywhere. 
\end{lem}
\begin{proof}
    For any $f \in H_c^2(\reals^n)$, we have $d_2 d_1 f = 0$, then $d_1 f \in \ker(d_2)$. Hence, by the definition of inner products on $C_1$, we have
    \begin{align*}
        0 = \ & \langle d_1 f, X \rangle_{_{1}} \\
        = \ & \sum^n_{i=1} \int_{\reals^n} \frac{\partial f}{\partial x_i} \cdot g_i dx + \sum^n_{i=1}\sum^n_{j=1} \int_{\reals^n}  \frac{\partial^2f}{\partial x_j \partial x_i} \cdot \frac{\partial g_i}{\partial x_j } dx \\
        = \ & - \left(\sum^n_{i=1} \int_{\reals^n} f \cdot \frac{\partial g_i }{\partial x_i} dx  + \sum^n_{i=1}\sum^n_{j=1} \int_{\reals^n}  \frac{\partial f}{\partial x_j } \cdot \frac{\partial^2 g_i}{\partial x_j \partial x_i} dx \right)\\
        = \ &  - \left(\int_{\reals^n} f \cdot \Div(X)  + \sum^n_{j=1} \int_{\reals^n}  \frac{\partial f}{\partial x_j } \cdot \frac{\partial\Div(X)}{\partial x_j } dx \right)\,.
    \end{align*}
    As $\Div(X) \in H_c^2(\reals^n)$, substituting it with $f$, we have $\Div(X) = 0$ almost everywhere. 
\end{proof}

\begin{lem}\label{lem:divDU_0}
    For $u \in \gV \oplus \gN$, and each $u_i \in C_c^\infty(\reals^n), \forall i \in [M]$, then $\Div(Du) = 0$. 
\end{lem}
\begin{proof}
    The result follows from the definition of $\gV$, $\gN$, and Lemma \ref{lem:div0_ae}. 
\end{proof}

we refer to the class of games with divergence-free individual gradient vector fields as vector potential games. However, we formally term the games in $\gV \oplus \gN $ as near vector potential games, as we have Lemma \ref{lem:divDU_0} under the assumption that each $u_i \in C_c^\infty(\reals^n)$. In the general case, when each $u_i$ does not have compact support, we only have the guarantee of 
\begin{align*}
    \Div(Du) + \sum^n_{j=1} \frac{\partial^2 \Div(Du)}{\partial^2 x_j} = 0\,.
\end{align*}
Note that when $u$ is compactly supported, i.e. each $u_i$ is compactly supported, then we can have $\Div(Du) = 0$, which represents an exact vector potential game. In Section \ref{sec:alterdecom}, we introduce a different decomposition scheme that gives an exact vector potential game.

While the learning dynamics are well understood in scalar potential games, their behaviors remain unclear in vector potential games. We consider the dynamical system defined by ascending according to the direction of the gradients, i.e. 
\begin{align}\label{eq:GD}
    \dot{x} = Du \,.
\end{align}

We first provide an example where gradient descent dynamics is orbiting in a game with zero divergence of the gradient vector field. 
\begin{example}\label{example:orbit}
    Consider a two-player game where player $1$ plays strategy $x \in \reals$ and player $2$ plays strategy $y \in \reals$. The utility functions are given as 
    \begin{align*}
        u_1(x,y) = xy \,, \quad u_2(x,y) = -xy - x^3 y\,.
    \end{align*}
    One can check that the gradient is $Du = (y, -x-x^3)$ and $\Div(Du) = 0$. Moreover, the orbits of the gradient descent are solutions to the curves
    \begin{align*}
        \frac{y^2}{2} + \frac{x^2}{2} + \frac{x^4}{4} = c^2 \,,
    \end{align*}
    where $c$ is a constant. 
\end{example}
Figure \ref{fig:orbit} visually illustrates the orbiting behaviors of gradient descent in a two-player game described Example \ref{example:orbit}. 
\begin{figure*}[t]
    \centering
    \includegraphics[width=\linewidth]{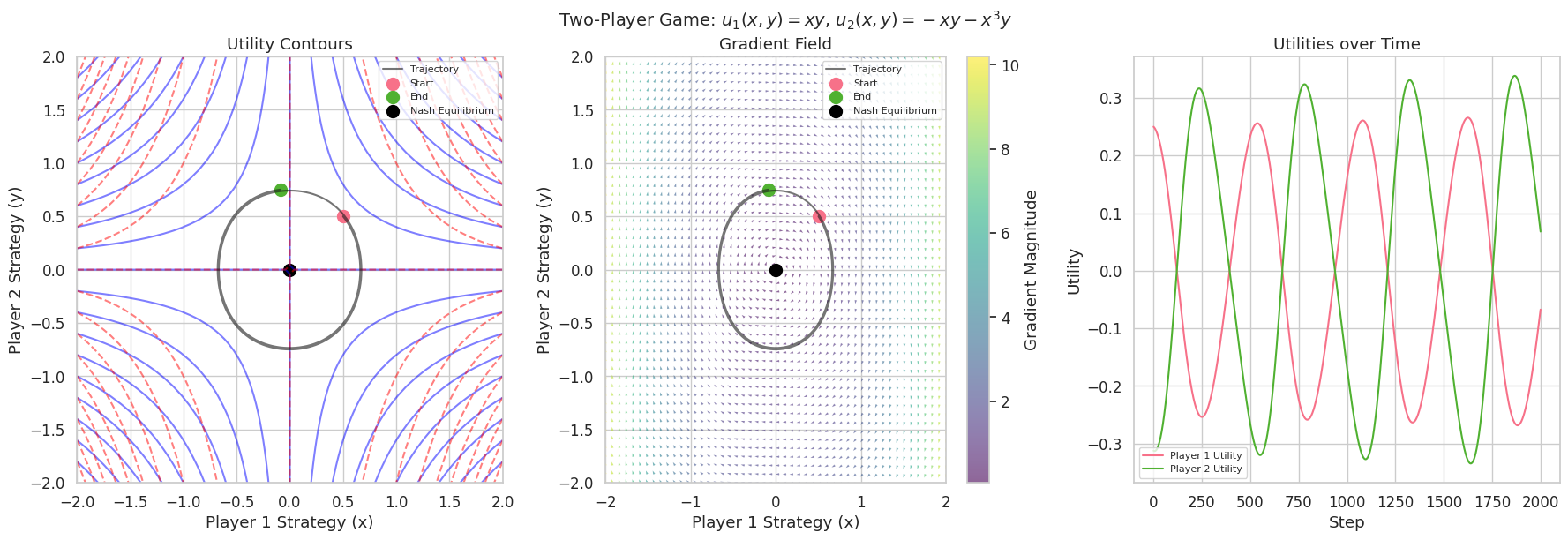}
    \caption{The dynamics of gradient descent, changes of utility over time and gradient vector field of the game presented in Example \ref{example:orbit}.}
    \label{fig:orbit}
\end{figure*}

To investigate the behavior of this dynamic system described above. We start with some definitions of flows on $\reals^n$. We refer to \cite{lee2012smooth} for a general theory and \cite{legacci2024geometric} (Appendix A) for a more comprehensive introduction.

Consider a smooth vector $X$ on $\reals^n$, a smooth global integral curve of $X$ is a smooth curve $x: \reals \rightarrow \reals^n$ such that $\dot{x}(t) = X(x(t)), \forall t \in \reals$. If a smooth global integral $x$ with starting point $y$ exists, then it is the unique maximal solution of 
\begin{align*}
    x(0) = y \,, \quad \dot{x}(t) = X(x(t))\,. 
\end{align*}
A smooth global flow on $\reals^n$ is a smooth map $\theta: \reals \times \reals^n \rightarrow \reals^n$ such that $\forall t, s \in \reals$ and $x \in \reals^n$, $\theta(0, x) = x$ and $\theta(t, \theta(s, x)) = \theta(t+s, x)$. Fix some $t \in \reals$, the orbit map $\theta_t: \reals^n \rightarrow \reals^n$ is $\theta_t(x) = \theta(t,x)$. Fix $x \in \reals$, the curve $\theta^x: \reals \rightarrow \reals^n$ denotes $\theta^x(t) = \theta(t, x)$. 

For any open set $\gU \subseteq \reals^n$ and $t \in \reals$, $\gU_t$ is defined to be image of $\gU$ under the orbit map $\theta_t$, i.e. $$\gU_t = \theta_t(\gU) = \{\theta_t(x): x \in \gU\} \subseteq \reals^n\,.$$ 
\begin{thm}[Euclidean Liouville's theorem]\label{thm:liouville}
Given a smooth vector field $X$ in $\reals^n$ and an open set $\gU \subseteq \reals^n$, $$\frac{d}{dt} \Vol(\gU_t) = \int_{\gU_t} \Div(X) dx \,,$$
$\forall t \in \reals$ such that the flow of $X$ is defined. 
\end{thm}

We say a map $\phi: \reals^n \rightarrow \reals^n$ is volume preserving if $\Vol(\gU) = \Vol(\phi \gU)$, $\forall \gU \subseteq \reals^n$. An important implication of Liouville's theorem is that orbit maps of vector fields with zero divergence are volume-preserving. 
\begin{cor}
    If a vector field $X$ in $\reals^n$ has $\Div(X) = 0$, then $\Vol(\gU_t) = \Vol(\gU), \forall \gU \subset \reals^n, \forall t \in \reals$ such that the flow of X is defined 
\end{cor}

To discuss the behaviors of the dynamical system defined in $X$, we maintain the following assumption of the dynamical system. We first remark that this assumption is required from an ergodic theory perspective \cite{bekka2000ergodic}, and then provide an example to illustrate the necessity of it.
\begin{asmp}[Bounded orbit]\label{asmp:bounded_orbit}
    There exists a compact set $\Omega$ such that for any $x_0 \in \Omega$, $\theta_{t}(x_0) \in \Omega$, $\forall t \geq 0$.
\end{asmp}
We now explain the necessity and implication of this assumption. Consider the case that there exists some $t_1$ such that $\theta_{t_1}(x_0) \notin \Omega$, then we consider a set $\Omega_1$, $\Omega \subseteq \Omega_1$, $2\Vol(\Omega) \leq \Vol(\Omega_1)$, $\theta_{t_1} (x_0) \in \Omega_1$. If there exists some $t_2 > t_1$, $\theta_{t_2}(x_0) \notin \Omega_1$. Then we consider some set $\Omega_2$, $\Omega_1 \subseteq \Omega_2$, $2\Vol(\Omega_1) \leq \Vol(\Omega_2)$, $\theta_{t_2} (x_0) \in \Omega_2$. By repeatedly applying this argument, we can see that the trajectory $x(t)$ induced by $x_0$ can go to infinity, and the dynamical system can be considered divergent. 

Without this assumption, a simple game where the trajectory of gradient descent goes to infinity is a two-player game with $Du = (y,x)$. We visually illustrate the dynamic in this game in Figure \ref{fig:infinite}. 
\begin{figure*}[t]
    \centering
    \includegraphics[width=\linewidth]{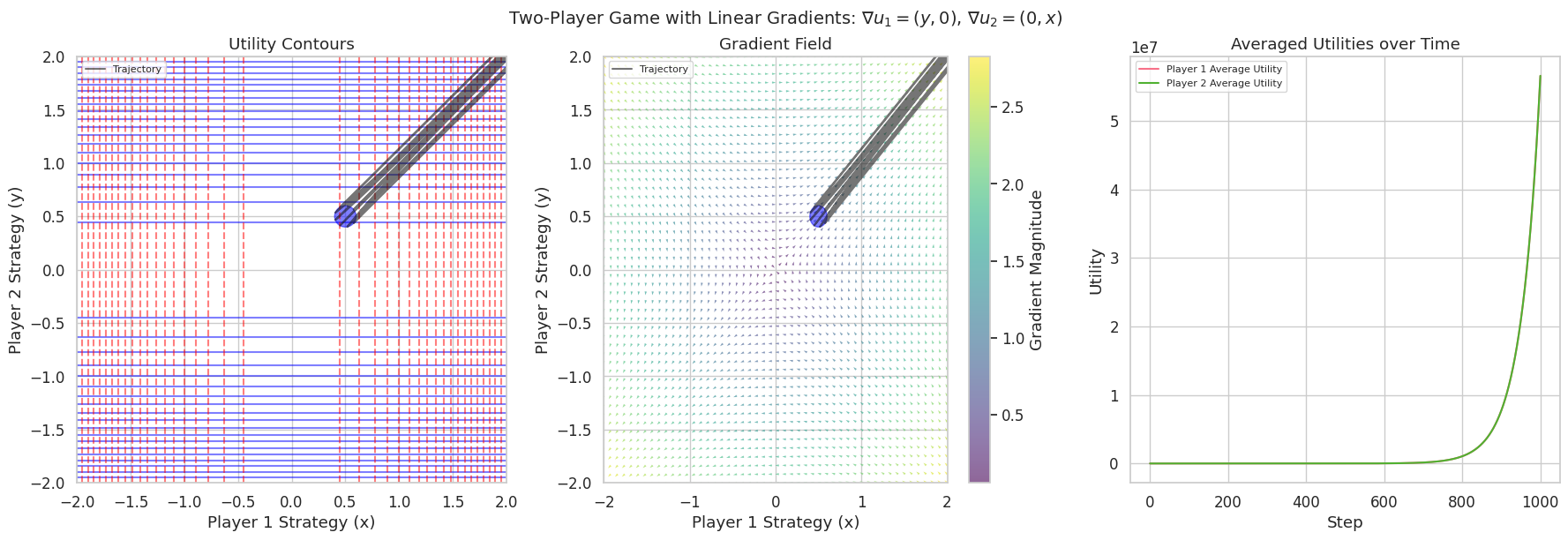}
    \caption{The dynamics of gradient descent, changes of utility over time, and gradient vector field of a two-player game with utility gradients $Du = (y,x)$. We plotted the trajectory that started from a set of initial points. We can see that the volume of the set remains unchanged throughout the trajectory. }
    \label{fig:infinite} 
\end{figure*}

\begin{lem}
    If $Du$ is compactly supported on $\Omega$, then Assumption \ref{asmp:bounded_orbit} is satisfied.
\end{lem}
\begin{proof}
    Suppose that the claim does not hold, then there exists an $x \in \Omega$, $t_1 > 0$, $x_1 = \theta^x(t_1) = \theta_{t_1}(x) \not\in \Omega$, where $\Omega$ is a compact set. Let $\tilde{t} = \sup\{t \mid \theta^x(t) \in \Omega, t \leq t_1\}$, $\tilde{x} = \theta^x(\tilde{t})$. Since $\theta^x(t)$ is continuous, we have $\tilde{x} \in \Omega$. By the mean value theorem, there exists a $t_2 \in (\tilde{t}, t_1)$ such that $Du(\theta^x(t_2)) = \frac{d \theta^x}{dt}(t_2) = \frac{x_1 - \tilde{x}}{t_1 - \tilde{t}} \neq 0$. By the definition of $\tilde{t}$, $\theta^x(t_2) \not\in \Omega$. Then we have $Du(\theta^x(t_2)) = 0$, which indicates that the claim holds by proof by contradiction. 
\end{proof}

Given a measure space $(\Omega, \mu)$, we say that $(\Omega, \mu)$ is finite if $\mu(\Omega)<\infty$, and that a map $\phi: \Omega \rightarrow \Omega$ is measure preserving if $\mu(\phi \mathcal{U})=\mu(\mathcal{U})$ for all measurable subsets $\mathcal{U} \subseteq \Omega$. 
\begin{thm}[Poincar\'e - Measure setting \cite{bekka2000ergodic}]\label{thm:poincare}
    Let $(\Omega, \mu)$ be a finite measure space, and let $\phi: \Omega \rightarrow \Omega$ be a measure preserving mapping. Let $\mathcal{U}$ be a measurable subset of $\Omega$. Then almost every point $x \in \mathcal{U}$ is infinitely recurrent with respect to $\mathcal{U}$, that is, the set $\left\{n \in \mathbb{N}: \phi^n x \in \mathcal{U}\right\}$ is infinite.
\end{thm}

Under our assumption, we obtain the following Lemma. 
\begin{lem}
    Let $(\Omega, \mu)$ be a Lebesgue measure space. Then, under Assumption \ref{asmp:bounded_orbit}, $\theta_t$ is measure-preserving. 
\end{lem}
\begin{proof}
    By Assumption \ref{asmp:bounded_orbit}, $\forall \gU \in \Omega$, $\tilde{\gU}_t = \gU_t \cap \Omega = \gU_t$. Then, by Theorem \ref{thm:liouville} $\mu(\tilde{\gU}_t) = \mu(\gU)$. 
\end{proof}

\begin{thm}\label{thm:poincare_near_vector}
    If a game is near vector potential, i.e. $Du \in \gV \oplus \gN$, and $u$ is compactly supported on $\Omega$, i.e. each $u_i$ is compactly supported, $i \in [M]$, then the system defined in Equation (\ref{eq:GD}) is a poincar\'e recurrent. Specifically, for almost every initialization $x(0) \in \Omega$, the induced trajectory $x(t)$ returns arbitrarily close to $x(0)$ infinitely often.
\end{thm}
\begin{proof}
    If $u$ has compact support and $u \in \gV \oplus \gN$, so does $Du$. Then $Du$ satisfies Assumption \ref{asmp:bounded_orbit}. By Lemma \ref{lem:divDU_0} and Theorem \ref{thm:liouville}, \ref{thm:poincare}, the statement follows. 
\end{proof}

\subsection{Connection to harmonic and Hamiltonian games}
To illustrate the relationship between the vector potential games and the Hamiltonian games, we first provide an example of a vector potential game that is not Hamiltonian. 
\begin{example}[Vector Potential games may not be Hamiltonian]
    Consider a two-player game where player $1$ plays strategy $x \in \reals$ and player $2$ plays strategy $y \in \reals$. The utility functions are given as 
    \begin{align*}
        u_1 (x, y) = x^2 \,, \quad u_2(x,y) = -y^2 \,.
    \end{align*}
    In this case, the gradient field is $Du = (2x, -2y)$ and therefore $\Div(Du) = 0$, so the game is a vector potential game. However, the Jacobian matrix $\begin{pmatrix}
2 & 0 \\
0 & -2 
\end{pmatrix}$ is not skew-symmetric and is therefore not Hamiltonian. 
\end{example}
However, we provide the following lemma to show that Hamiltonian games are vector potential games. 
\begin{lem}
    A Hamiltonian game is a vector potential game.
\end{lem}
\begin{proof}
    If a game with utility $u$ is Hamiltonian, then the diagonal of the Jacobian matrix is $0$. This implies $\Div(Du) = 0$ by the definition of divergence. 
\end{proof}

For a finite game, through Helmholtz decomposition, it can be decomposed into a potential part, which is the same as our scalar potential part, and a harmonic part. In differential games, however, the harmonic part does not exist. This is because of the fact harmonic component is isomorphic to the de Rham cohomology of the manifold, which is zero when the differential $k$-form is with $k = 1$ and the manifold is $\reals^n$. Despite the absence of harmonic games, the vector potential game of the differential games exhibits properties similar to those of harmonic games, such as the divergence-free gradient vector field. In fact, this is the key property that leads to the chaotic behavior of learning algorithms in harmonic games. 

From the decomposition point of view, the relationship between the vector potential game and the harmonic game can be illustrated by Figure \ref{fig:decomposition}. 
\begin{figure}
    \centering
    \includegraphics[width=0.6\linewidth]{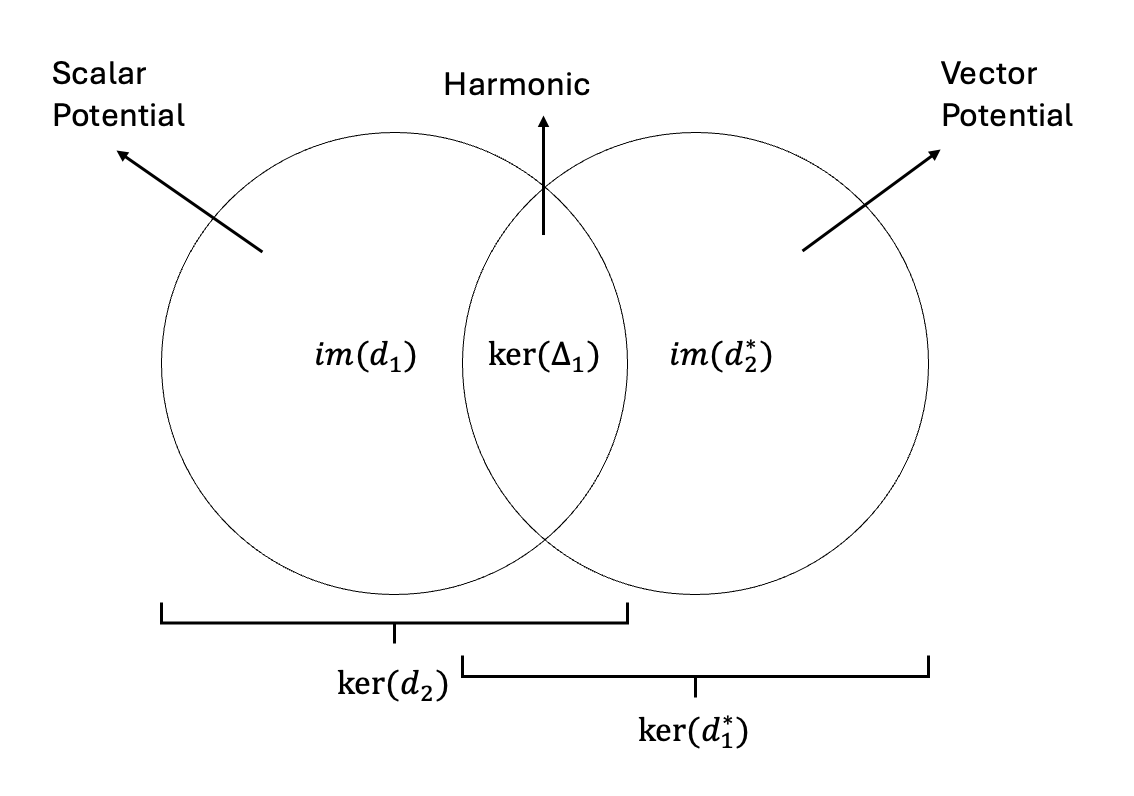}
    \caption{Decomposition of games. The vector Laplacian is defined as $\Delta_1 = d_2^* \circ d_2+d_1 \circ d_1^*,$}
    \label{fig:decomposition}
\end{figure}
In differential games on $\reals^n$, the harmonic part is 
\begin{align*}
    \ker(\Delta_1) = \ker(d_2) \cap \ker(d_1^\ast) = 0\,,
\end{align*}
as $\ima(d_1) =  \ker(d_2)$.


\subsection{Alternative decomposition}
\label{sec:alterdecom}
In the remarks following Lemma \ref{lem:divDU_0}, we mentioned that the decomposition presented so far only yields a direct sum decomposition of an exact scalar potential game and a near vector potential game. To bridge the gap between near and exact vector potential games, we operated under the assumption that the utility function is compactly supported to show the non-convergent behaviors in the vector potential game. In this section, we provide an alternative decomposition, which gives a decomposition of a near scalar potential game and an exact vector potential game. As we decompose to an exact vector potential game, we can show the non-convergent behaviors in the vector potential game without additional assumptions. One can choose either one, out of our two decompositions, according to their specific game. 

To obtain this decomposition, we define a slightly different $C_1$. Specifically, we use the following $\tilde{C}_1$. 
\begin{defn}\label{defn:tildeC1}
    The Hilbert space $\tilde{C}_1$ is defined as $$\tilde{C}_1 = \left\{X = (f_1, \ldots, f_n) \mid f_i \in L^2(\reals^n)\right\}\,.$$ The inner product is defined as 
    \begin{align*}
        \langle X, Y \rangle_{_{1}} = \sum^n_{i=1} \int_{\reals^n} f_i \cdot g_i dx\,. 
    \end{align*}
    The norm is thus $\|X\|_{_{1}} = \left[ \langle X, X \rangle_1\right]^{1/2}$.
\end{defn}
We then define the operator that maps between $C_0$ and $\tilde{C}_1$ as $\tilde{d}_1$. 
\begin{defn}\label{defn:tilde_d_1}
    The operator $\tilde{d}_1: C_0 \rightarrow \tilde{C}_1$ is defined as follows. For $f \in C_0$, $$\tilde{d}_1 f = \left(\frac{\partial f}{\partial x_1}, \ldots, \frac{\partial f}{\partial x_n}\right)\,.$$
\end{defn}

In this decomposition, we drop the condition that the utility has to be compactly supported. Instead, we assume that each $u_m$ is in the following class of $\tilde{C}_0^M$, where $$\tilde{C}_0^M = \{u = (u_1, \ldots, u_M) \mid u_i \in C_0, \forall i\}\,.$$ Using the same definition of operate $D$ and $D_m$, we obtain the following lemma. 
\begin{lem}\label{lem:Du_in_tilde_C1}
For any $m \in [M]$, $Du \in \tilde{C}_1$ and $(D u)_m \in L^2(\reals^n)$.
\end{lem}
\begin{proof}
The proof is similar to that of \ref{lem:Du_in_C1}, which we omit here. 
\end{proof}

Equipped with these results, we show the second decomposition.
\begin{figure*}[htb]
    \centering
    \includegraphics[width=\linewidth]{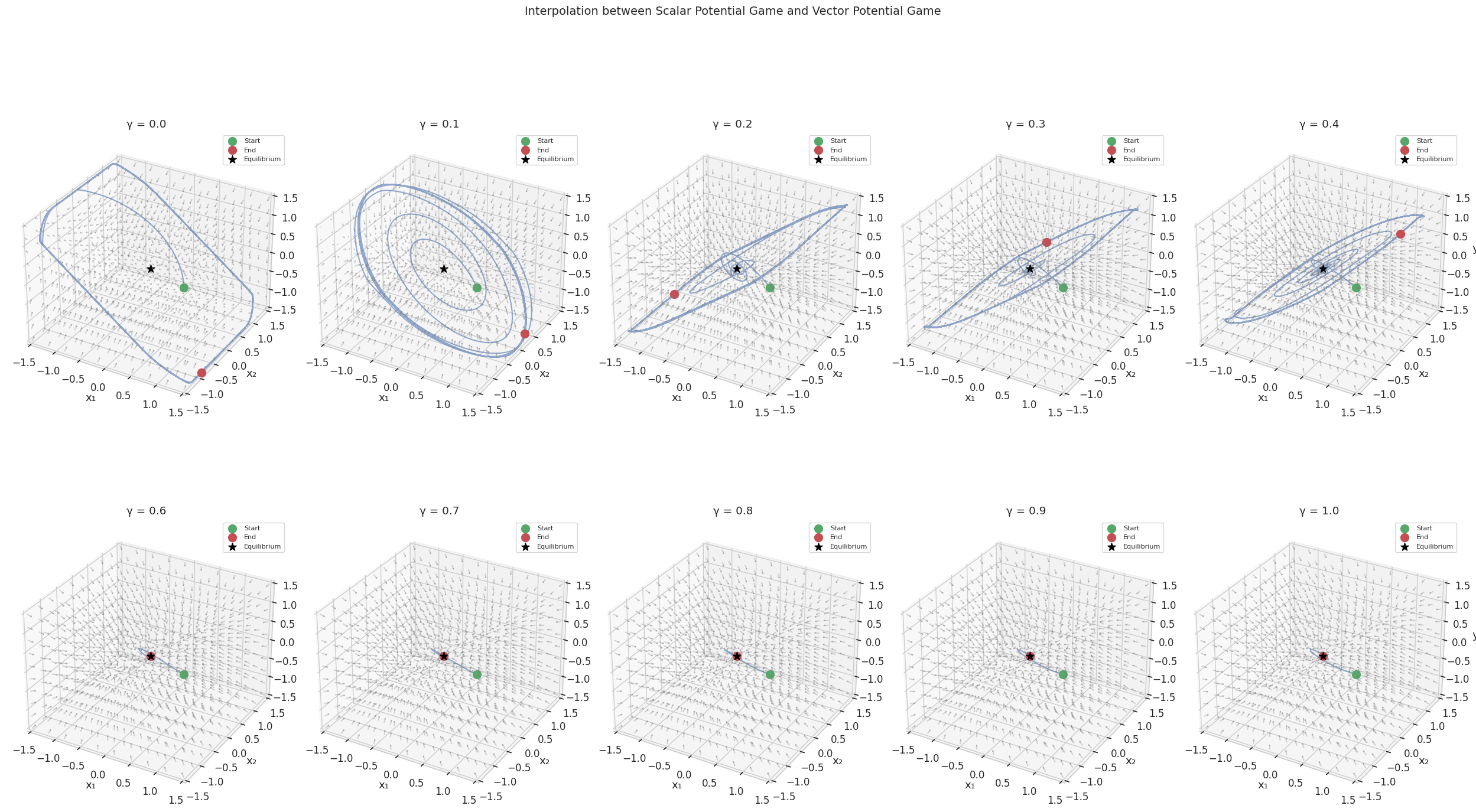}
    \caption{The strategies $x, y \in \reals^2$, $x = (x_1, x_2), y = (y_1, y_2)$ in a spectrum of games by interpolating scalar and vector potential games. The utility functions for the scalar potential game are
$
U_1 = -((x_1 - y_1)^2 + (x_2 - y_2)^2) - 0.2(x_1^2 + x_2^2) \,,
U_2 = -((x_1 - y_1)^2 + (x_2 - y_2)^2) - 0.2(y_1^2 + y_2^2)$, while the utility for the vector potential games are
$
U_1 = x_1 y_2 - x_2 y_1 
U_2 = -(x_1 y_2 - x_2 y_1) - 0.1(x_1^2 y_2^2 + x_2^2 y_1^2)$.}
\end{figure*}

\begin{thm}\label{thm:decompose_C1_tilde}
    $\tilde{C}_1$ can be decomposed as $$\tilde{C}_1 = \overline{\ima(\tilde{d}_1)} \oplus \ima(\tilde{d}_1)^\perp \,.$$ Then $$Du = X_{\tilde{\gP}} + X_{\tilde{\gV}}\,,$$ where $X_{\tilde{\gP}} \in \overline{\ima(\tilde{d}_1)}$, and $X_{\tilde{\gV}} \in \ima(\tilde{d}_1)^\perp$. 
\end{thm}
\begin{proof}
    The result follows immediately from the decomposition of $\tilde{C}_1$ and Lemma \ref{lem:Du_in_tilde_C1}.
\end{proof}

Now we can define the following classes of subgames, which can be interpreted as the space of the near scalar potential games, the exact vector potential games, and non-strategic games. 
\begin{align*}
    \tilde{\gP} = \ & \left\{u \in \tilde{C}_0^M \mid Du \neq 0\,, Du \in \overline{\ima(\tilde{d}_1)} \right\} \cup \{0\}\,, \\
    \tilde{\gV} = \ & \left\{u \in \tilde{C}_0^M \mid Du \neq 0\,, Du \in \ima(\tilde{d}_1)^\perp\right\}\cup \{0\} \,, \\
    \tilde{\gN} = \ & \left\{u \in C_0^M \mid u \in \ker(D)\right\} \,.
\end{align*}

We say the games in $\tilde{\gP} \oplus \tilde{\gN}$ are near potential games, as we can only find an $\epsilon$-potential function. Specifically, we have the following corollary for the near potential games. 
\begin{cor}
 If $u \in \tilde{\gP} \oplus \tilde{\gN}$, then for any $\epsilon > 0$, there exist some $\phi \in C_0$, such that $\|\tilde{d}_1\phi - Du\|_{_1} \leq \epsilon$.    
\end{cor}
The learning dynamics of finite near potential games have been extensively investigated in \cite{candogan2010projection,candogan2013dynamics,candogan2013near}. While it is unclear whether the results on the finite near potential game can be immediately extended to differential near potential games, it is reasonable to conjecture that they might behave similarly to that of an exact potential game. We leave it as future work to investigate the dynamics of differential near potential games.

For the exact vector potential game, we show that the divergence of the gradient field is exactly zero. 
\begin{lem}\label{lem:divDU_0_vector}
    If $X \in \ima (\tilde{d}_1)^\perp$, $X = (g_1, \ldots, g_n)$, and $g_i \in H^1(\reals^n)$ for each $i$, then $\Div(X) = 0$ almost everywhere. 
\end{lem}
\begin{proof}
    For any $f \in C_c^\infty(\reals^n)$, as $X \in \ima (\tilde{d}_1)^\perp$, $\langle\tilde{d}_1 f, X \rangle_{_1} = 0$. Then 
    \begin{align*}
        0 
        = \ & \sum^n_{i=1} \int_{\reals^n} \frac{\partial f}{\partial x_i} \cdot g_i dx \\
        = \ & \sum^n_{i=1} \left(-\int_{\reals^n} f \cdot \frac{\partial g_i}{\partial x_i} \right)\\
        =\ & - \int_{\reals^n} f \cdot \Div(X) dx \,.
    \end{align*}
    By Lemma \ref{lem:axuiliary_1}, we have $\Div(X) = 0$ almost everywhere.
\end{proof}

\begin{lem}\label{lem:axuiliary_1}
For any $f \in C_C^{\infty}\left(R^n\right)$, $g \in L^2(\reals^n)$, if $\int_{R^n} f(x) g(x) d x=0 $, then $g(x)=0$ almost everywhere. 
\end{lem}

\begin{proof}
    By Corollary 4.2 of Book 2 in \cite{brezis2011functional},  $C_c^{\infty}\left(R^n\right)$ is dense in $L^2\left(R^n\right)$. There exists $g_k \in C_c^{\infty}\left(R^n\right)$ such that $\left\|g_k-g\right\|_{L^2\left(R^n\right)} \longrightarrow 0$. Therefore $\int_{R^n} g_k(x) g(x) d x=0$, and 
\begin{align*}
& \left|\int_{R^n} g(x) \cdot g(x) d x-\int_{R^n} g_k(x) \cdot g(x) d x\right| \\
 \leq \ & \int_{R^n}\left|g_k-g\right| \cdot|g| d x \\
 \leq \ & \left\|g_k-g\right\|_{L^2\left(R^n\right)} \cdot\| g\|_{L^2\left(R^n\right)} \xrightarrow{\text { as } k \rightarrow \infty} 0 \,.
\end{align*}
As $$\int_{R^n}|g(x)|^2 d x=\lim _{k \rightarrow \infty} \int_{R^n} g_k(x) g(x) d x=0\,,$$  we have $g(x)=0$ almost everywhere.
\end{proof}

As one might expect, the dynamical system (\ref{eq:GD}) by gradients ascending exhibits similar behavior in vector potential games as is in near vector potential games. Namely, by the divergence-free result, it is either poincar\'e recurrent or divergent. Notice that, though, in vector potential games it no longer requires a compactly supported utility function. 
\begin{thm}\label{thm:poincare_vector}
    If a game is a vector potential game, i.e. $Du \in \tilde{\gV} \oplus \tilde{\gN}$, and Assumption \ref{asmp:bounded_orbit} is satisfied, Equation (\ref{eq:GD}) is a poincar\'e recurrent. Specifically, for almost every initialization $x(0) \in \Omega$, the induced trajectory $x(t)$ returns arbitrarily close to $x(0)$ infinitely often.
\end{thm}
\begin{proof}
    By Lemma \ref{lem:divDU_0_vector} and Theorem \ref{thm:liouville}, \ref{thm:poincare}, we have the result. 
\end{proof}
Beyond the recurrence, the Nash equilibrium point is also nonconvex and nonconcave in exact vector potential games. Hence, this property shows new insights into the challenges of finding Nash equilibrium in differential games with second-order optimization methods. 
\begin{lem}
    If $\Div(Du) = 0$, and $p$ is a local Nash equilibrium in $\gU$, i.e. $u_m(p) \geq u_m(q^m, p^{-m})$, $\forall m, q^m \in \gU$, then $\frac{\partial^2 u_m}{\partial^2 \omega_i^m} (p) = 0$. 
\end{lem}
\begin{proof}
    By the definition, we have $\Div(Du) = \sum^M_{m=1} \sum^n_{i=1}\frac{\partial^2 u_m}{\partial \omega_i^m} = 0$, and each $\frac{\partial^2 u_m}{\partial^2 \omega_i^m} (p) \leq 0$ as $p$ is a local Nash equilibrium. Therefore each $\frac{\partial^2 u_m}{\partial^2 \omega_i^m}$ must be $0$.
\end{proof}

\section{Interpolation of scalar and vector potential games}
Our decomposition method can interpolate the differential games on a spectrum, where the exact scalar and vector potential games are at the two ends. In most of the applications, the differential game considered is a mixture of the two games. While it is hard to obtain theoretical guarantees of the learning dynamics for the games in the middle of the spectrum, we provide the following figure to show the behaviors of gradient descent on the interpolations of scalar and vector potential games through simulation. The utility of the interpolated game is a convex combination of the scalar and vector potential game, $u = \gamma u_{\text{SP}} + (1-\gamma)u_{\text{VP}}$, where $u_{\text{SP}}$ is the utility function of a scalar potential game and $u_{\text{VP}}$ is the utility function of a vector potential game. The parameter $\gamma$ varies from $0$ to $1$, allowing for a smooth transition between the two types of games.

\section{Conclusion and Future Works}
In this work, we have provided two decompositions of differential games through the Hodge/Helmholtz decomposition. In summary, we decompose a differential game into a scalar potential game and a vector potential game. We showed that the scalar potential game is an exact potential game introduced by \cite{monderer1996potential}, where the gradient descent dynamic can effectively find the Nash equilibrium. For the vector potential game, we showed that the individual gradient field is divergence-free, which means that the gradient descent dynamic may either diverge or exhibit recurrent behavior. Technically, we introduced the Sobolev space to ensure that the gradient vector field is a Hilbert space, thereby achieving a direct sum decomposition of the simultaneous gradient field levels in differential games.

There are a number of possible future directions for this work. The first question is whether there exists a direct sum decomposition of differential games that gives an exact scalar potential game and an exact vector potential game. We conjecture that this would be unlikely, if not impossible, though there is no trivial way to show the negative result. The second important question is whether we can leverage this decomposition to improve existing learning dynamics in differential games. Namely, with what variant of optimization methods can we handle the divergence-free component?

\bibliography{ref}

\newcommand{\etalchar}[1]{$^{#1}$}
\begin{thebibliography}{GPAM{\etalchar{+}}20}

\bibitem[APFS22]{anagnostides2022last}
Ioannis Anagnostides, Ioannis Panageas, Gabriele Farina, and Tuomas Sandholm.
\newblock On last-iterate convergence beyond zero-sum games.
\newblock In {\em International Conference on Machine Learning}, 2022.

\bibitem[APSV22]{abdou2022decomposition}
Joseph Abdou, Nikolaos Pnevmatikos, Marco Scarsini, and Xavier Venel.
\newblock Decomposition of games: some strategic considerations.
\newblock {\em Mathematics of Operations Research}, 47(1):176--208, 2022.

\bibitem[BM00]{bekka2000ergodic}
M~Bachir Bekka and Matthias Mayer.
\newblock {\em Ergodic theory and topological dynamics of group actions on homogeneous spaces}, volume 269.
\newblock Cambridge University Press, 2000.

\bibitem[Bre11]{brezis2011functional}
H~Brezis.
\newblock Functional analysis, sobolev spaces and partial differential equations, 2011.

\bibitem[CDT09]{chen2009settling}
Xi~Chen, Xiaotie Deng, and Shang-Hua Teng.
\newblock Settling the complexity of computing two-player {N}ash equilibria.
\newblock {\em Journal of the ACM (JACM)}, 56(3):1--57, 2009.

\bibitem[CL16]{chen2016generalized}
Po-An Chen and Chi-Jen Lu.
\newblock Generalized mirror descents in congestion games.
\newblock {\em Artificial Intelligence}, 241:217--243, 2016.

\bibitem[CLZQ16]{cheng2016decomposed}
Daizhan Cheng, Ting Liu, Kuize Zhang, and Hongsheng Qi.
\newblock On decomposed subspaces of finite games.
\newblock {\em IEEE Transactions on Automatic Control}, 61(11):3651--3656, 2016.

\bibitem[CMOP11]{candogan2011flows}
Ozan Candogan, Ishai Menache, Asuman Ozdaglar, and Pablo~A Parrilo.
\newblock Flows and decompositions of games: Harmonic and potential games.
\newblock {\em Mathematics of Operations Research}, 36(3):474--503, 2011.

\bibitem[CMS10]{cominetti2010payoff}
Roberto Cominetti, Emerson Melo, and Sylvain Sorin.
\newblock A payoff-based learning procedure and its application to traffic games.
\newblock {\em Games fand Economic Behavior}, 70(1):71--83, 2010.

\bibitem[COP10]{candogan2010projection}
Ozan Candogan, Asuman Ozdaglar, and Pablo~A Parrilo.
\newblock A projection framework for near-potential games.
\newblock In {\em 49th IEEE Conference on Decision and Control (CDC)}, pages 244--249. IEEE, 2010.

\bibitem[COP13a]{candogan2013dynamics}
Ozan Candogan, Asuman Ozdaglar, and Pablo~A Parrilo.
\newblock Dynamics in near-potential games.
\newblock {\em Games and Economic Behavior}, 82:66--90, 2013.

\bibitem[COP13b]{candogan2013near}
Ozan Candogan, Asuman Ozdaglar, and Pablo~A Parrilo.
\newblock Near-potential games: Geometry and dynamics.
\newblock {\em ACM Transactions on Economics and Computation (TEAC)}, 1(2):1--32, 2013.

\bibitem[CXFD22]{cuilearning}
Qiwen Cui, Zhihan Xiong, Maryam Fazel, and Simon~Shaolei Du.
\newblock Learning in congestion games with bandit feedback.
\newblock In {\em Advances in Neural Information Processing Systems}, 2022.

\bibitem[Das13]{daskalakis2013complexity}
Constantinos Daskalakis.
\newblock On the complexity of approximating a {N}ash equilibrium.
\newblock {\em ACM Transactions on Algorithms (TALG)}, 9(3):1--35, 2013.

\bibitem[DK21]{donahue2021optimality}
Kate Donahue and Jon Kleinberg.
\newblock Optimality and stability in federated learning: A game-theoretic approach.
\newblock {\em Advances in Neural Information Processing Systems}, 34:1287--1298, 2021.

\bibitem[Eva98]{evans1998partial}
LC~Evans.
\newblock " partial differential equations,''.
\newblock {\em Graduate Studies in Mathematics}, 1998.

\bibitem[GPAM{\etalchar{+}}20]{goodfellow2020generative}
Ian Goodfellow, Jean Pouget-Abadie, Mehdi Mirza, Bing Xu, David Warde-Farley, Sherjil Ozair, Aaron Courville, and Yoshua Bengio.
\newblock Generative adversarial networks.
\newblock {\em Communications of the ACM}, 63(11):139--144, 2020.

\bibitem[HCM17]{heliou2017learning}
Am{\'e}lie Heliou, Johanne Cohen, and Panayotis Mertikopoulos.
\newblock Learning with bandit feedback in potential games.
\newblock In {\em Advances in Neural Information Processing Systems}, 2017.

\bibitem[HMC03]{hart2003uncoupled}
Sergiu Hart and Andreu Mas-Colell.
\newblock Uncoupled dynamics do not lead to nash equilibrium.
\newblock {\em American Economic Review}, 93(5):1830--1836, 2003.

\bibitem[HMC06]{hart2006stochastic}
Sergiu Hart and Andreu Mas-Colell.
\newblock Stochastic uncoupled dynamics and nash equilibrium.
\newblock {\em Games and economic behavior}, 57(2):286--303, 2006.

\bibitem[KMA{\etalchar{+}}21]{kairouz2021advances}
Peter Kairouz, H~Brendan McMahan, Brendan Avent, Aur{\'e}lien Bellet, Mehdi Bennis, Arjun~Nitin Bhagoji, Kallista Bonawitz, Zachary Charles, Graham Cormode, Rachel Cummings, et~al.
\newblock Advances and open problems in federated learning.
\newblock {\em Foundations and trends{\textregistered} in machine learning}, 14(1--2):1--210, 2021.

\bibitem[LBR{\etalchar{+}}19]{letcher2019differentiable}
Alistair Letcher, David Balduzzi, S{\'e}bastien Racaniere, James Martens, Jakob Foerster, Karl Tuyls, and Thore Graepel.
\newblock Differentiable game mechanics.
\newblock {\em Journal of Machine Learning Research}, 20(84):1--40, 2019.

\bibitem[LCH19]{li2019note}
Changxi Li, Daizhan Cheng, and Fenghua He.
\newblock A note on orthogonal decomposition of finite games.
\newblock {\em arXiv preprint arXiv:1905.08112}, 2019.

\bibitem[Lit94]{littman1994markov}
Michael~L Littman.
\newblock Markov games as a framework for multi-agent reinforcement learning.
\newblock In {\em Machine learning proceedings 1994}, pages 157--163. Elsevier, 1994.

\bibitem[LL12]{lee2012smooth}
John~M Lee and John~M Lee.
\newblock {\em Smooth manifolds}.
\newblock Springer, 2012.

\bibitem[LMP24]{legacci2024geometric}
Davide Legacci, Panayotis Mertikopoulos, and Bary Pradelski.
\newblock A geometric decomposition of finite games: Convergence vs. recurrence under exponential weights.
\newblock In {\em ICML 2024-41st International Conference on Machine Learning}, 2024.

\bibitem[LWT{\etalchar{+}}17]{lowe2017multi}
Ryan Lowe, Yi~I Wu, Aviv Tamar, Jean Harb, OpenAI Pieter~Abbeel, and Igor Mordatch.
\newblock Multi-agent actor-critic for mixed cooperative-competitive environments.
\newblock {\em Advances in neural information processing systems}, 30, 2017.

\bibitem[MPP18]{mertikopoulos2018cycles}
Panayotis Mertikopoulos, Christos Papadimitriou, and Georgios Piliouras.
\newblock Cycles in adversarial regularized learning.
\newblock In {\em Proceedings of the twenty-ninth annual ACM-SIAM symposium on discrete algorithms}, pages 2703--2717. SIAM, 2018.

\bibitem[MS96]{monderer1996potential}
Dov Monderer and Lloyd~S. Shapley.
\newblock Potential games.
\newblock {\em Games and economic behavior}, 14(1):124--143, 1996.

\bibitem[PPP17]{palaiopanos2017multiplicative}
Gerasimos Palaiopanos, Ioannis Panageas, and Georgios Piliouras.
\newblock Multiplicative weights update with constant step-size in congestion games: Convergence, limit cycles and chaos.
\newblock {\em Advances in Neural Information Processing Systems}, 30, 2017.

\bibitem[PSV{\etalchar{+}}23]{panageas23}
Ioannis Panageas, Stratis Skoulakis, Luca Viano, Xiao Wang, and Volkan Cevher.
\newblock Semi bandit dynamics in congestion games: {C}onvergence to {N}ash equilibrium and no-regret guarantees.
\newblock In {\em International Conference on Machine Learning}, 2023.

\bibitem[Rou10]{roughgarden2010algorithmic}
Tim Roughgarden.
\newblock Algorithmic game theory.
\newblock {\em Communications of the ACM}, 53(7):78--86, 2010.

\bibitem[VGFP19]{vlatakis2019poincare}
Emmanouil-Vasileios Vlatakis-Gkaragkounis, Lampros Flokas, and Georgios Piliouras.
\newblock Poincar{\'e} recurrence, cycles and spurious equilibria in gradient-descent-ascent for non-convex non-concave zero-sum games.
\newblock {\em Advances in Neural Information Processing Systems}, 32, 2019.

\bibitem[WLC17]{wang2017weighted}
Yuanhua Wang, Ting Liu, and Daizhan Cheng.
\newblock From weighted potential game to weighted harmonic game.
\newblock {\em IET Control Theory \& Applications}, 11(13):2161--2169, 2017.

\end{thebibliography}

\end{document}